\documentclass[conference]{IEEEtran}

\makeatletter
\newcommand\semihuge{\@setfontsize\semihuge{22.3}{22}}
\makeatother
\usepackage[dvips]{color}
\usepackage{comment}
\usepackage{todonotes}
\usepackage{epsf}
\usepackage{epsfig}
\usepackage{times}
\usepackage{epsfig}
\usepackage{graphicx}
\usepackage{bbold}
\usepackage{mathtools}
\usepackage{mathrsfs}
\usepackage{amssymb}
\usepackage{pdfpages}
\usepackage{epstopdf}
\usepackage{dsfont}
\usepackage{lettrine} % \lettrine[findent=1pt]{{{R}}}{}
\usepackage{amsmath,epsfig,amssymb,algorithm,algpseudocode,amsthm,cite,url}
\usepackage{subcaption}
\allowdisplaybreaks
\usepackage{csquotes}
\usepackage{verbatim}
\usepackage[english]{babel}
\usepackage{amsmath,amssymb}

\usepackage[justification=centering]{caption}
\usepackage{verbatim}

\newtheorem{theorem}{\bf Theorem}

\newtheorem{proposition}{\bf Proposition}
\newtheorem{lemma}{\bf Lemma}

\begin{document}
	
\title{\Huge  Optimal Transport Theory for Cell Association in UAV-Enabled Cellular Networks \vspace{-0.15cm}}    

\author{ %\IEEEauthorblockN{Author 1, Author 2, Author 3, and Author 4 \vspace{-1.3cm}  }
\IEEEauthorblockN{  Mohammad Mozaffari$^1$, Walid Saad$^1$, Mehdi Bennis$^2$, and M\'erouane Debbah$^3$}\vspace{-0.1cm}\\
	\IEEEauthorblockA{
		\small $^1$ Wireless@VT, Electrical and Computer Engineering Department, Virginia Tech, VA, USA,\\ Emails:\url{{mmozaff , walids}@vt.edu}.\\
		$^2$ CWC - Centre for Wireless Communications, Oulu, Finland, Email: \url{bennis@ee.oulu.fi}.\\
		$^3$ Mathematical and Algorithmic Sciences Lab, Huawei France R \& D, Paris, France, and CentraleSup´elec,\\   Universit´e Paris-Saclay, Gif-sur-Yvette, France, Email: \url{merouane.debbah@huawei.com}.
%		%\thanks{This paper will be presented in part at the IEEE GLOBECOM conference, Washington, DC, USA, 2016 \cite{MozaffariIoT}.}
%		%\thanks{This work was supported by the U.S. National Science
%		%Foundation under Grants AST-1506297, by the Office of Naval Research (ONR) under Grant N00014-15-1-2709, and, by the ERC Starting
%		%Grant 305123 MORE (Advanced Mathematical Tools for Complex Network
%     	%Engineering), and by the Academy of Finland.}
	}\vspace{-0.62cm}
	}\vspace{-2.9cm}
\maketitle\vspace{-2.99cm}

\begin{abstract}\vspace{-0.00cm}
	%Unmanned areal vehicles (UAVs) can be used as areal base stations to boost the coverage and capacity of the terrestrial cellular networks. 
	\textcolor{black}{In this paper, a novel framework for delay-optimal cell association in unmanned aerial vehicle (UAV)-enabled cellular networks is proposed.} %\textcolor{black}{First, the maximum achievable coverage range of each UAV is derived as a function of altitude.} %, and the terrestrial base station (BS) is derived. %First, to maximize the UAVs' coverage range, the optimal altitude of each UAV is derived.
	 In particular, to minimize the average network delay under any arbitrary spatial distribution of the ground users, the optimal cell partitions of UAVs and terrestrial base stations (BSs) are determined. To this end, using the powerful mathematical tools of optimal transport theory, \textcolor{black}{the existence of the solution to the optimal cell association problem is proved and the solution space is completely characterized}. \textcolor{black}{The analytical and simulation results show that the proposed approach yields substantial improvements of the~average network~delay.}\vspace{-0.01cm}%Analytical results for specific setup parameters show that, using the proposed approach, the average network delay decreases 72\% compared to a classical signal strength-based association.\vspace{-0.2cm}}%deploying UAVs instead of terrestrial small BSs can lead to a 27\% network delay reduction.} %Furthermore, using the proposed approach, the average network delay decreases by a factor of 2.5 compared to a classical signal strength-based association}. 

\end{abstract} \vspace{0.1cm}

\section{Introduction}%\vspace{-0.01cm}

\textcolor{black}{The use of unmanned aerial vehicles (UAVs) such as drones and balloons is an effective technique for improving the quality-of-service (QoS) of wireless cellular networks due to their inherent ability to create line-of-sight (LoS) communication links\cite{ orfanus,mozaffari2, zhangLetter,zhang,Letter, Irem, HouraniModeling}.} %In particular, UAVs can be used as aerial base stations to enhance the capacity, coverage, and quality-of-service (QoS) of terrestrial wireless cellular networks \cite{mozaffari2, Irem, HouraniModeling}. Due to their mobility and inherent ability to establish line-of-sight (LoS) communication links, UAVs can provide wide-scale, reliable wireless communications in various scenarios. For instance, UAVs can be deployed to support cellular systems in emergency situations in which the terrestrial network is not fully operational. Moreover, during temporary events and in hotspots, UAVs can be used to provide additional wireless network capacity. %Compared to terrestrial base stations (BSs), UAVs can provide better coverage performance by establishing line-of-sight (LoS) links to ground users.
 Nevertheless, there are many technical challenges associated with the UAV-based communication systems, which include deployment, path planning, flight time constraints, and cell association. 
  In \cite{Letter} and \cite{Irem}, the authors studied the efficient deployment of aerial base stations to maximize the coverage performance. %The work in \cite{HouraniModeling} performed air-to-ground channel modeling for UAV-based communications. 
  The path planning challenge and optimal trajectory of UAVs were addressed in \cite{Jeong} and \cite{Qing}. Moreover, UAV communications under flight time considerations was studied in \cite{HoverTime}. 
%In \cite{mozaffari2}, we  investigated the downlink coverage and rate performance of a single UAV that co-exists with a device-to-device (D2D) communication network.
 %The authors in \cite{Vishnu} derived the downlink coverage probability while using multiple UAVs as aerial base stations. 
 %\vspace{-0.4cm}
\textcolor{black}{Another important challenge in UAV-based communications is cell (or user) association.} In \cite{Vishal}, the authors analyzed the user-UAV assignment for capacity enhancement of heterogeneous networks. However, this work is limited to the case in which users are uniformly distributed within a geographical area. %Furthermore, the work in \cite{Vishal} does not consider the network congestion resulting from the non-uniform spatial distribution of the users which is more practical for UAV scenarios. 
 \textcolor{black}{In \cite{Alonso}, the authors proposed a power-efficient cell association scheme while satisfying the rate requirement of users in cellular networks. However, in \cite{Alonso}, the authors do not consider the presence of UAVs and their objective function does not account for network delay. In \cite{OTUAV}, the optimal deployment and cell association of UAVs are determined with
the goal of minimizing the UAVs' transmit power while satisfying the users' rate requirements.
However, the work in \cite{OTUAV} mainly focused on the optimal deployment of the UAVs and does not analyze the existence and characterization of the cell association problem.} \textcolor{black}{Therefore, our work is different from \cite{OTUAV} in terms of the system model, the objective function, the problem formulation as well as analytical results.}
%While the work in \cite{Kims} studies the optimal cell association and load balancing for throughput maximization, it is limited to a terrestrial network. Furthermore, in \cite{Kims}, the impact of non-uniform spatial distribution of users on the optimal cell association was not taken into account. 
In fact, none of the previous studies in \cite{Letter, HoverTime,Qing,orfanus, mozaffari2, zhangLetter, HouraniModeling, Irem,Alonso, Vishal, Jeong,zhang,OTUAV}, addressed the delay-optimal cell association problem considering both UAVs and terrestrial base stations, for any arbitrary distribution~of~users.%\vspace{-0.3cm}

The main contribution of this paper is to introduce a novel framework for delay-optimal cell association in a cellular network in which both UAVs and terrestrial BSs co-exist. In particular, given the locations of the UAVs and terrestrial BSs as well as any general spatial distribution of users, we find the optimal cell association by exploiting the framework of \emph{optimal transport theory} \cite{villani}. \textcolor{black}{Within the framework of optimal transport theory, one can address cell association problems for any general spatial distribution of users. In fact, the main advantage of optimal transport theory is to provide tractable solutions for a variety of cell association problems in wireless networks. In our problem,} we first prove the  existence of the optimal solution to the cell association problem, and, then, we characterize the solution space. \textcolor{black}{The results show that, our approach results in a significantly lower delay compared to a conventional signal strength-based association.\vspace{-0.1cm}}% \vspace{-0.5cm} % Moreover, using the proposed cell association approach, the average network delay  can decrease by a factor of 2.5 compared to a conventional signal strength-based association.}%\vspace{-0.3cm} 

\section{System Model and Problem formulation}\vspace{-0.05cm} 

Consider a geographical area $\mathcal{D}\subset \mathds{R}^2$ in which  $K$ terrestrial BSs in set $\mathcal{K}$ are deployed to provide
service for ground users that are spatially distributed according to a distribution $f(x,y)$ over the two-dimensional plane. In
addition to the terrestrial BSs, $M$ UAVs in set $\mathcal{M}$ are deployed as aerial base stations to enhance
the capacity of the network. %\textcolor{blue}{We show the sets of BSs and UAVs indices, respectively, by $\mathcal{K}$ and $\mathcal{M}$}.
 We consider a downlink
scenario in which the BSs and the UAVs use a frequency division multiple access (FDMA) technique to service the ground users. \textcolor{black}{The locations of BS $i\in \mathcal{K}$ and UAV $j\in \mathcal{M}$ are, respectively, given by $(x_i,y_i,h_i)$ and $(x^\textrm{uav} _{j},y^\textrm{uav} _{j},h^\textrm{uav}_j)$, with $h_i$ and $h^\textrm{uav}_j$ being the heights of BS $i$ and UAV $j$.} %In our model, the altitudes of the UAVs are not fixed and they can change depending the on locations of the users.
The maximum transmit powers of BS $i$ and UAV $j$ are  $P_i$ and $P_j^\textrm{uav}$. Let $W_i$ and $W_j$ be the total bandwidth available for each BS $i$ and UAV $j$.  %Let $N_i$ and $N_j^\textrm{uav}$ be the number of users associated with BS $i$ and UAV $j$.% Hence, the transmit power used by each BS $i$ and UAV $j$ to serve each of their associated users are equal to $\frac{{{P_i}}}{{{N_i}}}$ and $\frac{{{P_j^\textrm{uav}}}}{{{N_j^\textrm{uav}}}}$.
\textcolor{black}{Our performance metric is the \textit{transmission delay}, which is referred to as the time needed for transmitting a given number of bits. In this case, the delay is inversely proportional to the transmission rate.} We use $A_i$ and $B_j$ to denote, respectively, the area (cell) partitions in which the ground users are assigned to BS $i$ and UAV $j$. Hence, the geographical area is divided into $M+K$ disjoint partitions each of which is served by one of the BSs or the UAVs. 
 
 Given this model, our goal is to minimize the average network delay by optimal partitioning of the area. %and adjusting the UAVs' altitudes. % Naturally, to maximize the network reliability, we need to minimize the average outage probability of the users.
  Based on the spatial distribution of the users,  we determine the optimal cell associations to minimize the average network delay. Note that, the network delay significantly depends on the cell partitions due to the following reasons. First, the cell partitions determine the service area of each UAV and BS thus impacting the channel gain that each user experiences. Second, the number of users in each partition depends on the cell partitioning. In this case, since the total bandwidth is limited, the amount of bandwidth per user decreases as the  number of users in a cell partition increases. \textcolor{black}{Thus, users in the crowded cell partitions achieve a lower throughput which results in a higher delay}. %To determine the optimal cell association, %as well as the altitude of the UAVs,
 Next, we present the channel models. \vspace{-0.1cm} % used for BS-user and UAV-user communications.\vspace{-0.2cm} 

   \begin{figure}[!t]
   	\begin{center}
   		\vspace{-0.1cm}
   		\includegraphics[width=7.2cm]{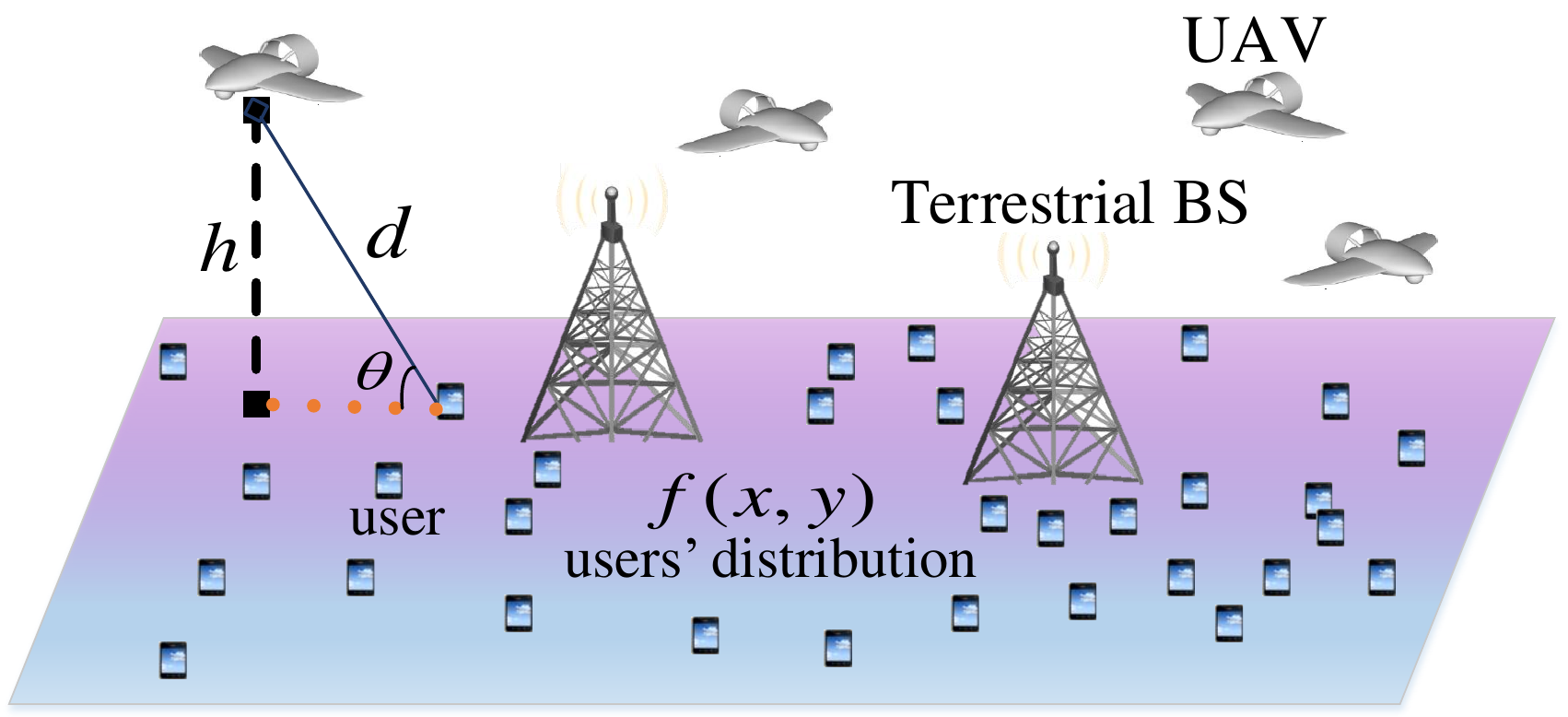}
   		\vspace{-0.01cm}
   		\caption{ \small Network model. }\vspace{-.52cm}
   		\label{SystemModel}
   	\end{center}
   \end{figure}    
   % Note that, the the network reliability significantly depends on the cell partitions due to the following reasons. First, cell partitions determine the serving area of each UAV and BS thus impacting the channel gain that each user experience. Second,  the outage probability is also a function of the number of users per partition which itself depends on the cell partitioning. For a higher number of users in a cell partition, the power limited BS or UAV must use a lower transmit power (per user) to serve each user. Therefore, users located in the crowded cell partitions may experience higher outage probabilities. To determine the optimal cell partitioning, we first need to present the channel models used for BS-user and UAV-user communications.    

\subsection{UAV-User and BS-User path loss models}\vspace{-0.05cm} 
\textcolor{black}{In UAV-to-ground communications, the probability of having LoS links to users depends on the locations, heights, and the number of obstacles, as well as the elevation angle between a given UAV and
it’s served ground user. In our model, we consider a commonly used probabilistic path loss model provided by International
Telecommunication Union (ITU-R), and the work in \cite{HouraniModeling}}.
%For the UAV-user path loss modeling, a common approach is to consider the LoS and non-line-of-sight (NLoS) links between the UAV and the ground users separately as detailed in \cite{HouraniModeling}. %and \cite{FengPath}.
%Each link has a specific probability of occurrence which depends on the elevation angle, environment, and relative location of the UAV and the users. %Clearly, for NLoS links the shadowing and blockage loss is higher than the LoS links.
The path loss between UAV $j$ and a user located at $(x,y)$ is \cite{HouraniModeling}:\vspace{-0.2cm}
\begin{equation}\label{Pr}
	\Lambda_j = \left\{\hspace{-0.16cm} \begin{array}{l}
		 K_o^{2}(d_j/d_o)^{2} {\mu _\text{LoS},} \,\,\,\hspace{0.35cm}{\text{LoS link,}}\\
		 K_o^{2}(d_j/d_o)^{2} {\mu _\text{NLoS},} \,\,\,\hspace{0.21cm}{\text{NLoS link,}}
	\end{array} \right. \vspace{-0.2cm} 
\end{equation}
%where $P_{r,j}^\textrm{uav}$ is the received signal power and $P_{j}^\textrm{uav}$ is the UAV's transmit power. 
where \begin{small}$K_o=\left(\frac{4\pi f_cd_o} {c}\right)^{2}$\end{small}, $f_c$ is the carrier frequency, $c$ is the speed of light, and $d_o$ is the free-space reference distance. Also, $\mu_\textrm{LoS}$ and $\mu_ \textrm{NLoS}$ are different attenuation factors considered for LoS and NLoS links. 
\begin{small}$d_j=\sqrt{(x-x^\textrm{uav} _{j})^2+(y-y^\textrm{uav} _{j})^2+h{^\textrm{uav}_j}^2}$\end{small} is the distance between UAV $j$ and an arbitrary ground user located at $(x,y)$. For the UAV-user link, the LoS probability is \cite{HouraniModeling}:\vspace{-0.2cm} 
\begin{equation} \label{PLoS}
	{\mathds{P}_{\text{LoS},j}} = \alpha {\left( {\frac{180}{\pi}\theta_j  - 15} \right)^\gamma}, \,\,\, \theta_j>\frac{\pi}{12},\vspace{-0.15cm}
\end{equation}
 where $\theta_j={\sin ^{- 1}}( \frac {h_j} {d_j})$ is the elevation angle (in radians) between the UAV and the ground user. Also, $\alpha$ and $\gamma$ are constant values reflecting the environment impact. Note that, the NLoS probability is  $\mathds{P}_{\text{NLoS},j}=1-\mathds{P}_{\text{LoS},j}$.

Considering $d_o=1$\,m, the average path loss is $K_o{ {{d_j}}^{ 2}}\left[ {{\mathds{P}_{\textrm{LoS},j}}{\mu _\textrm{LoS}} + {\mathds{P}_{\textrm{NLoS},j}}{\mu _\textrm{NLoS}}} \right]$. \textcolor{black}{Therefore, the received signal power from UAV $j$ considering an equal power allocation among its associated users will be:\vspace{-0.2cm}
\begin{equation}\label{Puav_ave}
\hspace{-0.01cm} \bar P_{r,j}^\textrm{uav} = P_j^\textrm{uav}/\left( N_j^\textrm{uav} K_o{ {{d_j}}^{2}}\left[ {{\mathds{P}_{\textrm{LoS},j}}{\mu _\textrm{LoS}} + {\mathds{P}_{\textrm{NLoS},j}}{\mu _\textrm{NLoS}}} \right]\right),\hspace{-0.3cm} \vspace{-0.2cm} 
\end{equation}
where $P_{j}^\textrm{uav}$ is the UAV's total transmit power, and $N_j^\textrm{uav} =N \iint_{{B_j}} {f(x,y)\textrm{d}x\textrm{d}y}$ is the average number of users associated with UAV $j$, with $N$ being the total number of users.}
For the BS-user link, we use the traditional path loss model. %\footnote{\textcolor{black}{Note that, for the UAV-user link when $\theta_j \le \frac{\pi}{12}$, we use the same traditional path loss model as for the BS-user link. In this case, the received signal power form UAV $j$ will be 	$P_{r,j}^\textrm{uav}=	P_j^\textrm{uav} K_o^{-1} d_j^{-{n}}/N_j^\textrm{uav}$.}}.
  \textcolor{black}{ In this case, the received signal power from BS $i$ at user's location $(x,y)$ will be:\vspace{-0.3cm}
\begin{equation}
	{P_{r,i}}=	{P_i} K_o ^{-1} d_i^{-{n}}/N_i, \vspace{-0.2cm} 
\end{equation}
where $d_i=\sqrt{(x-x_i)^2+(y-y_i)^2+h_i^2}$ is the distance between BS $i$ and a given user, $N_i =N \iint_{{A_i}} {f(x,y)\textrm{d}x\textrm{d}y}$ is the average number of users associated with BS $i$, and $n$ is the path loss exponent.}\vspace{-0.1cm}% for the BS-user link. 

 \subsection{Problem formulation}
 
 \textcolor{black}{Given the average received signal power in the UAV-user communication, the average throughput of a user located at $(x,y)$ connecting to a UAV $j$ can be approximated by:}\vspace{-0.2cm}
 \begin{align}
 &C_j^\textrm{uav}= {\frac{{W_j}}{{N_j^\textrm{uav}}}{{\log }_2}\big( {1 + \frac{\bar P_{r,j}^\textrm{uav}}{\sigma^2}} \big)},\vspace{-0.1cm} %\nonumber\\&= \frac{{W_j^\textrm{uav}}}{{N_j^\textrm{uav}}}{\log _2}\left( {1 + \frac{P_j^\textrm{uav} {{K_o}}{ {{d_j}}^{ - 2}}\left[ {{P_{\textrm{LoS},j}}{\mu _\textrm{LoS}} + {P_{\textrm{NLoS},j}}{\mu _\textrm{NLoS}}} \right]}{{\sigma^2}}} \right)
 \end{align}
\textcolor{black}{where $\sigma^2$ is the noise power for each user which is linearly proportional to the bandwidth allocated to the user.}
 
 \textcolor{black}{The throughput of the user if it connects to a BS $i$ is:\vspace{-0.1cm} 
 \begin{equation}
 {C_i} = \frac{{{W_i}}}{{{N_i}}}{\log _2}\big( {1 + \frac{{{P_{r,i}}}}{{  \sigma^2}}} \big).\vspace{-0.2cm}
 \end{equation}}
 Now, let $\mathcal{L}=\mathcal{K} \cup \mathcal{M}$ be the set of all BSs and UAVs. Also, here, the location of each BS or UAV is denoted by $\boldsymbol{s}_k$, $k\in\mathcal{L}$. We also consider ${D_k} = \left\{ \hspace{-0.15cm}\begin{array}{l}
 {A_k},\hspace{0.3cm} \textrm{if}\,\,k \in \mathcal{K},\\
 {B_k}, \hspace{0.3cm} \textrm{if}\,\, k \in \mathcal{M},
 \end{array} \right.$ denoting all the cell partitions, and \textcolor{black}{$
 Q\left( {\boldsymbol{v},\boldsymbol{s}_k,D_k} \right) = \left\{ \hspace{-0.15cm}\begin{array}{l}
b/ {C_k},\hspace{0.5cm} \textrm{if}\,\,k \in \mathcal{K},\\
b/ C_{k}^\textrm{uav},\hspace{0.3cm} \textrm{if}\,\,k \in \mathcal{M},
 \end{array} \right.$ where $\boldsymbol{v}=(x,y)$ is the 2D locations of the ground users, and $b$ is the number of bits that must be transmitted to location $\boldsymbol{v}$.} Then, our optimization problem that seeks to minimize the average network delay over the entire area will be: \vspace{-0.1cm}
 \begin{align} \label{Opt1}
 &\mathop {\min }\limits_{{D_k}} \sum\limits_{k \in \mathcal{L}} {\int_{{D_k}} {Q\left( {\boldsymbol{v},\boldsymbol{s}_k,{D_k}} \right)f(x,y)\textrm{d}x\textrm{d}y} }, \\
 %\textrm{s.t.}\,\,  & \left\| {\boldsymbol{v} - \boldsymbol{s}_k} \right\|^2-h_k^2\le (R_k^\textrm{max})^2,\,\,\, {\textrm{if} \, \boldsymbol{v}\in D_k}, \,\,\, \forall k\in \mathcal{L}, \label{Rcons} \\
\textrm{s.t.}\,\,  &\bigcup\limits_{k \in \mathcal{L}} {{D_k}}  = \mathcal{D},\,\,\,{D_l} \cap {D_m} = \emptyset ,\,\,\,\forall l \ne m \in \mathcal{L}. \vspace{-0.5cm} \label{Union} 
 \end{align} 
 \textcolor{black}{where both constraints in (\ref{Union})  guarantee that the cell partitions are disjoint and their union covers the entire area, $\mathcal{D}$.}
% Next, we propose a solution to (\ref{Opt1}) that leads to the optimal cell partitions, $D_k, \,\, \forall k\in \mathcal{L}$.
  %\vspace{-0.1cm}
% where (\ref{Rcons}) is the necessary condition for connecting each user to BS or UAV $k$. 

\section{Optimal Transport Theory for Cell Association}

Given the locations of the BSs and the UAVs as well as the distribution of the ground users, we find the optimal cell association for which the average delay of the network is minimized. 
\textcolor{black}{Let \textcolor{black}{${g_k}(z) = \frac{Nz}{{W_k}}$}, with $W_k$ being the bandwidth for each BS or UAV $k$ and $z$ is a generic argument.} %Note that, $g_k$ captures the impact of cell congestion on the network throughput.
\textcolor{black}{Also, we consider:\vspace{-0.2cm}
%\begin{small}
\begin{equation}
F(\boldsymbol{v},\boldsymbol{s}_k) = \left\{ \hspace{-0.2cm} \begin{array}{l}
b/{\log _2}\left(1+ P_{r,k}(\boldsymbol{v},\boldsymbol{s}_k)/\sigma^2 \right),\,\, \textrm{if}\,\, k\in\mathcal{K},\\
b/{\log _2}\left(1+ {\bar P_{r,j}^\textrm{uav}(\boldsymbol{v},\boldsymbol{s}_k)}/\sigma^2 \right),\,\, \textrm{if}\,\,  k\in\mathcal{M}.
\end{array} \right.\hspace{-0.1cm} \vspace{-0.2cm}
\end{equation}}
%\end{small}

%\begin{small}
%\begin{equation}
%F(\boldsymbol{v},\boldsymbol{s}_k) = \left\{ \hspace{-0.2cm} \begin{array}{l}
%{\log _2}\left(1+ {{P_k}{K_o}{{\left\| {\boldsymbol{v} - \boldsymbol{s}_k} \right\|}^{ - {n}}}}/\sigma^2 \right), k\in\mathcal{K},\\
%{\log _2}\left(1+ {G(\boldsymbol{v},\boldsymbol{s}_k)P_k^\textrm{uav}{K_o}{{ {{{\left\| {\boldsymbol{v} - \boldsymbol{s}_k} \right\|}}{{}}} }^{ -2}}}/\sigma^2 \right) , k\in\mathcal{M},
%\end{array} \right.
%\end{equation}
%\end{small}where $G(\boldsymbol{v}, \boldsymbol{s}_k)={P_\textrm{LoS}}\left(\boldsymbol{v}, \boldsymbol{s}_k \right) \mu _\textrm{LoS} + {P_\textrm{NLoS}}\left( {\boldsymbol{v}, \boldsymbol{s}_k} \right){\mu _\textrm{NLoS}}$, and $\left\| {.} \right\|$ is the Euclidean norm.  
Now, the optimization problem in (\ref{Opt1}) can be rewritten as:
\begin{align} \label{Opt3}
&\hspace{-0.9cm}\mathop {\min }\limits_{{D_k}} \begin{small} \sum\limits_{k \in \mathcal{L}} {\int_{{D_k}} {\left[ {{g_k}\left( {\int_{{D_k}} {f(x,y)\textrm{d}x\textrm{d}y} } \right)F(\boldsymbol{v},\boldsymbol{s}_k)} \right]f(x,y)\textrm{d}x\textrm{d}y} }\end{small},\\
%\textrm{s.t.}\,\,  & \left\| {\boldsymbol{v} - \boldsymbol{s}_k} \right\|^2-h_k^2\le (R_k^\textrm{max})^2,\,\,\, {\textrm{if} \, \boldsymbol{v}\in D_k}, \,\,\, \forall k\in \mathcal{L},  \\
\textrm{s.t.}\,\,  & \bigcup\limits_{k \in \mathcal{L}} {{D_k}}  = \mathcal{D},\,\,\,{D_l} \cap {D_m} = \emptyset ,\,\,\,\forall l \ne m \in \mathcal{L},
\end{align} 
where $D_k$ is the cell partition of each BS or UAV $k$.

Solving the optimization problem in (\ref{Opt3}) is challenging and intractable due to various reasons. First, the optimization variables $D_k$, $ \forall k \in \mathcal{L}$, are sets of continuous partitions which are mutually dependent. Second, $f(x,y)$ can be any generic function of $x$ and $y$ that leads to the complexity of the given two-fold integrations. To overcome these challenges, next, we model this problem by exploiting \emph{optimal transport theory} \cite{villani} in order to characterize the solution.    
%Next, we tackle this problem using the optimal transport theory.

%\subsection{Optimal Transport Theory}
% Optimal transport theory goes back to the Monge's problem in 1781 which is stated as follows \cite{villani}. Given piles of sands and holes with the same volume, what is the best move (transport map) to entirely fill up the holes with the minimum total transportation cost. % In its general form, optimal transport theory deals with finding an optimal transportation plan between two sets of points that leads to a minimum transportation cost \cite{villani2003}.

%\subsection{Optimal Transport Theory}
% Optimal transport theory goes back to the Monge's problem in 1781 which is stated as follows \cite{villani}. Given piles of sands and holes with the same volume, what is the best move (transport map) to entirely fill up the holes with the minimum total transportation cost. % In its general form, optimal transport theory deals with finding an optimal transportation plan between two sets of points that leads to a minimum transportation cost \cite{villani2003}.

% In general, this theory aims to find the optimal matching between two sets of points considering the costs associated with the matching between sets. These sets can be either discrete or continuous, with arbitrary distributions (weights).
Optimal transport theory \cite{villani} allows analyzing complex problems in which, for two probability measures $f_1$ and $f_2$ on $\Omega \subset \mathds{R}^n$, one must find the optimal transport map $T$ from $f_1$ to $f_2$  that minimizes the following function:\vspace{-0.12cm}
 \begin{equation}
 {\mathop {\min }\limits_T \int_\Omega {c\left( {x,T(x)} \right)} f_1(x)\textrm{d}x;\,\,T:\Omega \to \Omega}, \vspace{-0.05cm}
 \end{equation}
 where $c(x,T(x))$ denotes the cost of transporting a unit mass from a location $x$ to a location $T(x)$.
 
 \textcolor{black}{Our cell association problem can be modeled as a semi-discrete optimal transport problem. In this case, the users follow a continuous distribution, and the base stations can be considered as discrete points. Then, we need to map the users to the BSs and UAVs such that the total cost function is minimized. In this case,  the optimal cell partitions are directly determined by the optimal transport map \cite{Crippa}. Next, we prove the existence of the optimal solution to the problem in (\ref{Opt3}).} \vspace{-0.1cm}  
  %In general, this problem does not necessarily have a solution  as each point must be mapped to only one location. However, Kantorovich relaxed this problem by using transport plans instead of maps, in which one point can go to multiple points \cite{villani2003}.  as mass splitting is not allowed in the Monge problem. In other words, sands located in the same location must be transported to a same destination. Instead,  Kantorovich relaxed this problem by allowing mass splitting and using transport plans instead of the transport maps (in a map each point goes to one point but in a plan one point can go to multiple points). It is shown that an optimal transport plan always exists.
\begin{theorem}
	\normalfont
\textcolor{black}{The optimization problem in (\ref{Opt3}) admits an optimal solution given $N\ne0$, and $\sigma\ne 0$.\vspace{-0.2cm}}
\end {theorem}
\begin{proof}
Let ${a_k} = \int_{{D_k}} {f(x,y)\textrm{d}x\textrm{d}y}$, \textcolor{black}{and for $\forall k \in \mathcal{L}$,}\\
\begin{small}{\hspace{-0.3cm}$E \hspace{-0.2cm}=\hspace{-0.2cm} \left\{ {\boldsymbol{a} = \left( {{a_1},{a_2},...,{a_{K + M}}} \right) \in {\mathds{R}^{K + M}};{a_k} \ge 0,  \sum\limits_{k = 1}^{K + M} {{a_k} = 1} } \right\}$}\end{small}.
Now, considering $f(x,y)=f(\boldsymbol{v})$ and $c\left( {\boldsymbol{v},{\boldsymbol{s}_k}} \right) = {g_k}({a_k})F\left( {\boldsymbol{v},{\boldsymbol{s}_k}} \right)$, for any given vector $\boldsymbol{a}$, problem (\ref{Opt3}) can be considered as a classical semi-discrete optimal transport problem. 
%\begin{equation} \label{Opt4}
%\mathop {\min }\limits_T \int_\mathcal{D} {c\left( {\boldsymbol{v},\boldsymbol{s}} \right)} f(\boldsymbol{v})\textrm{d}\boldsymbol{v}, \,\, \boldsymbol{s}=T(\boldsymbol{v}),
%\end{equation}
%where $T$ is the transport map which is associated to cell partitions $D_k$ by:\vspace{-0.15cm}
%\begin{equation}
%\left\{T(v) = \sum\limits_{k \in \mathds{N}} {{s_k}{\mathds{1}_{{D_k}}}(v)}; {a_k} = \int_{{D_k}} {f(x,y)\textrm{d}x\textrm{d}y}\right\}.
%\end{equation}
%
First, we prove that $c\left( {\boldsymbol{v},\boldsymbol{s}} \right)$ is a semi-continuous function. Considering the fact that $\boldsymbol{s}_k$ is discrete, we have:
$\mathop {\lim }\limits_{(\boldsymbol{v},\boldsymbol{s}) \to ({\boldsymbol{v}^*},{\boldsymbol{s}_k})} F\left( {\boldsymbol{v},\boldsymbol{s}} \right)\mathop  =  \mathop {\lim }\limits_{\boldsymbol{v} \to {\boldsymbol{v}^*}} F\left( {\boldsymbol{v},{\boldsymbol{s}_k}} \right)$.
Note that, given any $\boldsymbol{s}_k$, $k$ belongs to only of $\mathcal{K}$ and $\mathcal{M}$ sets. Given $\boldsymbol{s}_k$, $F(\boldsymbol{v},\boldsymbol{s}_k)$ is a continuous function of $\boldsymbol{v}$. Then, considering the fact that given $a_k$, $g_k(a_k)$ is constant, we have $\mathop {\lim }\limits_{(\boldsymbol{v},\boldsymbol{s}) \to ({\boldsymbol{v}^*},{\boldsymbol{s}_k})} g_k(a_k) F\left( {\boldsymbol{v},\boldsymbol{s}} \right) = g_k(a_k)F\left( {{\boldsymbol{v}^*},{\boldsymbol{s}_k}} \right)$. Therefore,  $c(\boldsymbol{v},\boldsymbol{s})$ is a continuous function and, hence, is also a lower semi-continuous function.
 Now, we use the following lemma from optimal transport theory:\vspace{-0.15cm} %to show the existence of the optimal solution to (\ref{Opt3}).  
\begin {lemma}
\normalfont
Consider two probability measures $f$ and $\lambda$ on $\mathcal{D} \subset \mathds{R}^n$. Let $f$ be continuous and $\lambda = \sum\limits_{k \in \mathds{N}} {{a_k}{\delta _{{\boldsymbol{s}_k}}}}$ be a discrete probability measure. Then, for any lower semi-continuous cost function, there exists an optimal transport map from $f$ to $\lambda$ for which $\int_\mathcal{D} {c\left( {x,T(x)} \right)} f(x)\textrm{d}x$ is minimized \cite{Crippa}. \vspace{-0.1cm}  
\end{lemma}
\noindent Considering Lemma 1, for any $\boldsymbol{a}\in E$, the problem in (\ref{Opt3}) admits an optimal solution. Since $E$ is a unit simplex in $\mathds{R}^{M+K}$ which is a non-empty and compact set, the problem admits an optimal solution over the entire $E$.
\end{proof}

%\subsection{ Solution Characterization}
Next, we characterize the solution space of (\ref{Opt3}). %which allows finding the optimal cell partitioning.
\begin{theorem}
	\normalfont
	To acheive the delay-optimal cell partitions in (\ref{Opt3}), each user located at $(x,y)$ must be assigned to the following BS (or UAV): \vspace{-0.3cm}  
	%The optimal cell association rule for that leads to the optimal cell partitions in problem (\ref{Opt2}), can be given by:
\textcolor{black}{
\begin{small}	
\begin{equation}\label{The2}
	k=\mathop {\arg \min }\limits_{l \in \mathcal{L}} \big\{\frac{{{a_l}}}{{{W_l}}}F({\boldsymbol{v}_o},{\boldsymbol{s}_l})\big\},
\end{equation}
\end{small}	\vspace{-0.35cm}\\
 Given (\ref{The2}), the optimal cell partition $D_{k}$ includes all the points which are assigned to BS (or UAV) $k$.}	
%	\begin{align} \label{The2}
%&D_k=\biggl\{(x,y)|k=\mathop {\arg \min }\limits_{l \in \mathcal{K}} \nonumber\\
%&\left({g'_l\left( {{a_l}} \right) \int_{{D_l}} F({v},{s_l})f(x,y)\textrm{d}x\textrm{d}y}  + {g_l}\left( {{a_l}} \right)F({v},{s_l})\right)\biggr\}
%	\end{align}
\end{theorem}
\begin{proof}
\textcolor{black}{ As proved in Theorem 1, there exist optimal cell partitions $D_k$, $k\in \mathcal{L}$ which are the solutions to (\ref{Opt3}).} Now, consider two partitions $D_l$ and $D_m$, and a point $\boldsymbol{v}_o=(x_o,y_o)\in D_l$. Also, let $B_\epsilon(\boldsymbol{v}_o)$ be a ball with a center $\boldsymbol{v}_o$ and radius $\epsilon >0$. \textcolor{black}{Now, we generate the following new cell partitions ${{\mathord{\buildrel{\lower3pt\hbox{$\scriptscriptstyle\frown$}}\over D} }_k}$ (which are variants of the optimal partitions):} \vspace{-0.1cm}
\begin{small}
\begin{equation}
\left\{ \begin{array}{l}
{{\mathord{\buildrel{\lower3pt\hbox{$\scriptscriptstyle\frown$}} 
			\over D} }_l} = D_l\backslash {B_\varepsilon }({\boldsymbol{v}_o}),\\
{{\mathord{\buildrel{\lower3pt\hbox{$\scriptscriptstyle\frown$}} 
			\over D} }_m} = D_m \cup {B_\varepsilon }({\boldsymbol{v}_o}),\\
{{\mathord{\buildrel{\lower3pt\hbox{$\scriptscriptstyle\frown$}} 
			\over D} }_k} = D_k,\,\,\,\,k \ne l,m.
\end{array} \right.
\end{equation}
\end{small}
Let ${a_\varepsilon } = \int_{{B_\varepsilon }({\boldsymbol{v}_o})} {f(x,y)\textrm{d}x\textrm{d}y} $, and ${{\mathord{\buildrel{\lower3pt\hbox{$\scriptscriptstyle\frown$}} 
			\over a} }_k} = \int_{{{\mathord{\buildrel{\lower3pt\hbox{$\scriptscriptstyle\frown$}} 
				\over D} }_k}} {f(x,y)\textrm{d}x\textrm{d}y}$. Considering the optimality of $D_k$, $k\in \mathcal{L}$, we have:
\begin{small}
\begin{align}
&\hspace{0.1cm}\sum\limits_{k \in \mathcal{K}} {\int_{{D_k}} { {{g_k}\left( {{a_k}} \right)F(\boldsymbol{v},{\boldsymbol{s}_k})} f(x,y)\textrm{d}x\textrm{d}y} } \nonumber \\
& {\mathop  \le \limits^{(a)} }\sum\limits_{k \in \mathcal{K}} {\int_{{{\mathord{\buildrel{\lower3pt\hbox{$\scriptscriptstyle\frown$}} 
					\over D} }_k}} { {{g_k}\left( {{{\mathord{\buildrel{\lower3pt\hbox{$\scriptscriptstyle\frown$}} 
							\over a} }_k}} \right)F(\boldsymbol{v},{\boldsymbol{s}_k})} f(x,y)\textrm{d}x\textrm{d}y} }. \label{SUM}\\
&\textrm{\textcolor{black} {\normalsize Now, canceling out the common terms in (\ref{SUM}) leads to:}} \nonumber \\
\hspace{0.1cm}&\int_{{D_l}} { {{g_l}\left( {{a_l}} \right)F(\boldsymbol{v},{\boldsymbol{s}_l}) } f(x,y)\textrm{d}x\textrm{d}y}  + \int_{{D_m}} { {{g_m}\left( {{a_m}} \right)F(\boldsymbol{v},{\boldsymbol{s}_m})} f(x,y)\textrm{d}x\textrm{d}y}\nonumber\\
& \le \int_{{D_m} \cup {B_\varepsilon }({\boldsymbol{v}_o})} { {{g_m}\left( {{a_m} + {a_\varepsilon }} \right)F(\boldsymbol{v},\boldsymbol{s}_m)} f(x,y)\textrm{d}x\textrm{d}y} \nonumber\\ &+\int_{{D_l}\backslash {B_\varepsilon }({\boldsymbol{v}_o})} { {{g_l}\left( {{a_l} - {a_\varepsilon }} \right)F(\boldsymbol{v},\boldsymbol{s}_l)} f(x,y)\textrm{d}x\textrm{d}y},\nonumber\\
&\int_{{D_l}} { {\left( {{g_l}\left( {{a_l}} \right) - {g_l}\left( {{a_l} - {a_\varepsilon }} \right)} \right)F(\boldsymbol{v},\boldsymbol{s}_l)} f(x,y)\textrm{d}x\textrm{d}y}\nonumber\\ & +\int_{{B_\varepsilon }({\boldsymbol{v}_o})} { {{g_l}\left( {{a_l} - {a_\varepsilon }} \right)F(\boldsymbol{v},\boldsymbol{s}_l)} f(x,y)\textrm{d}x\textrm{d}y}\nonumber\\
&\le \int_{{D_m}} { {\left( {{g_m}\left( {{a_m}}+ {a_\varepsilon } \right) - {g_m}\left( {{a_m}} \right)} \right)F(\boldsymbol{v},\boldsymbol{s}_m)} f(x,y)\textrm{d}x\textrm{d}y}\nonumber\\
& + \int_{{B_\varepsilon }({\boldsymbol{v}_o})} { {{g_m}\left( {{a_m} + {a_\varepsilon }} \right)F(\boldsymbol{v},{\boldsymbol{s}_m})} f(x,y)\textrm{d}x\textrm{d}y},\label{ineq}\vspace{-0.15cm}
\end {align} \vspace{-0.13cm}
\end{small}\\
\noindent\textcolor{black}{where $(a)$ comes from the fact that $D_k$, $\forall k \in \mathcal{L}$ are optimal and, hence, any variation of such optimal partitions, shown by ${{{\mathord{\buildrel{\lower3pt\hbox{$\scriptscriptstyle\frown$}} \over D} }_k}}$, cannot lead to a better solution. Now, we multiply both sides of the inequality in (\ref{ineq}) by $\frac{1}{a_\epsilon}$, take the limit when $\epsilon \to 0$, and use the following equalities:\vspace{-0.2cm}
\begin{small}
\begin{align}
&\mathop {\lim }\limits_{\varepsilon  \to 0} {a_\varepsilon } = 0,\label{C1}\\
&\mathop {\lim }\limits_{{a_\varepsilon } \to 0} \frac{{{g_l}({a_l}) - {g_l}({a_l} - {a_\varepsilon })}}{{{a_\varepsilon }}} = {g'_l}({a_l}),\label{C2}\\
&\mathop {\lim }\limits_{{a_\varepsilon } \to 0} \frac{{{g_m}({a_m} + {a_\varepsilon }) - {g_m}({a_m})}}{{{a_\varepsilon }}} = {g'_m}({a_m}),\label{C3}\vspace{-0.3cm}
\end{align}
\end{small}
then we have:\vspace{-0.2cm}
% and considering $\mathop {\lim }\limits_{\varepsilon  \to 0} \frac{1}{{{\varepsilon ^2}}}\int\limits_{{y_o}}^{{y_o} + \varepsilon } {\int\limits_{{x_o}}^{{x_o} + \varepsilon } {f(x,y)\textrm{d}x\textrm{d}y} }  = f({x_o},{y_o})$, we have:
\begin{small}
\begin{align}
%&\hspace{-0.1cm}{f}(x_o,y_o)\int_{{D_l}} { {g'_l\left( {{a_l}} \right)F(\boldsymbol{v},{\boldsymbol{s}_l})} f(x,y)\textrm{d}x\textrm{d}y}  + {f}(x_o,y_o) {{g_l}\left( {{a_l}} \right)F(\boldsymbol{v},{\boldsymbol{s}_l})} \nonumber\\
%&\le {f}(x_o,y_o)\int_{{D_m}} { {g'_m\left( {{a_m}} \right)F(\boldsymbol{v},{\boldsymbol{s}_m})} f(x,y)\textrm{d}x\textrm{d}y} \nonumber\\
%& \hspace{0.3cm} + {f}(x_o,y_o) {{g_m}\left( {{a_m}} \right)F(\boldsymbol{v},{\boldsymbol{s}_m})}, \nonumber\\
&{g'_l\left( {{a_l}} \right)\int_{{D_l}} \hspace{-0.2cm}F(\boldsymbol{v}_o,{\boldsymbol{s}_l})f(x,y)\textrm{d}x\textrm{d}y}+ {g_l}\left( {{a_l}} \right)F({\boldsymbol{v}_o},{\boldsymbol{s}_l})\nonumber\\
&\hspace{-0.1cm}\le {g'_m\left( {{a_m}} \right) \int_{{D_m}} \hspace{-0.3cm}F({\boldsymbol{v}_o},\boldsymbol{s}_m)f(x,y)\textrm{d}x\textrm{d}y}  + {g_m}\left( {{a_m}} \right)\hspace{-0.05cm}F(\boldsymbol{v}_o,\boldsymbol{s}_m).\hspace{-0.10cm} 
\end{align}  
\end{small}
Now, given ${g_k}(z) = \frac{Nz}{{W_k}}$, we can compute ${g'_l}({a_l}) = {\left. {\frac{{d{g_l}(z)}}{{dz}}} \right|_{z = {a_l}}}=\frac{N}{{W_k}}$, then, using ${a_k} \hspace{-0.05cm}=\hspace{-0.05cm} \int_{{D_k}} {f(x,y)\textrm{d}x\textrm{d}y}$ leads~to:
\begin{small}
	\begin{align}
&\frac{N}{{{W_l}}}{a_l}F({\boldsymbol{v}_o},{\boldsymbol{s}_l}) + \frac{{N{a_l}}}{{{W_l}}}F({\boldsymbol{v}_o},{\boldsymbol{s}_l})\nonumber\\ 
&\le \frac{N}{{{W_m}}}{a_m}F({\boldsymbol{v}_o},{\boldsymbol{s}_m}) + \frac{{N{a_m}}}{{{W_m}}}F({\boldsymbol{v}_o},{\boldsymbol{s}_m}),\nonumber
\end{align}
\begin{equation}
\textrm{as a result: }\frac{{{a_l}}}{{{W_l}}}F({\boldsymbol{v}_o},{\boldsymbol{s}_l}) \le \frac{{{a_m}}}{{{W_m}}}F({\boldsymbol{v}_o},{\boldsymbol{s}_m}).\label{proof} 
\end{equation}
\end{small}  
Finally, (\ref{proof}) leads to (\ref{The2}) that completes the proof.}
\end{proof}

Theorem 2 provides a precise cell association rule for ground users that are distributed following any general distribution $f(x,y)$. In fact, the inequality given in (\ref{proof}) captures the condition under which the user is assigned to a BS or UAV $l$. Under the special case of a uniform distribution of the users, the result in Theorem 2 leads to the classical SNR-based association in which users are assigned to base stations that provide strongest signal.  %Furthermore, in a special case in which all base stations are identical, the proposed association results in the classical Voronoi diagram. \vspace{-0.1cm} 
\textcolor{black}{From Theorem 2, we can see that there is a mutual dependence between $a_l$ and $D_l$ (i.e. cell association), $\forall l\in \mathcal{L}$. \textcolor{black}{To solve the equation given in Theorem 2, we adopt an iterative approach which is shown to converge to the global optimal solution \cite{Crippa}}. In this case, we start with initial cell partitions (e.g. Voronoi diagram), and iteratively update the cell partitions based on Theorem 2.\vspace{-0.00cm}} 
%As a result, the solution of this cell association problem does not have an explicit form and, hence, we have adopt an iterative-based approach to solve the equation given in Theorem 2. In particular, we use a fixed-point iteration algorithm which is shown to converge to the optimal solution \cite{Alons}. In this case, we start with initial cell partitions (e.g. Voronoi diagram), and then iteratively update the cell partitions based on the result of Theorem 2. Therefore, $a_l$ is updated in each iteration.

%\vspace{0.2cm}

\section{Simulation Results and Analysis}\vspace{-0.00cm}
For our simulations, we consider an area of size $4\,\text{km}\times 4 \,\text{km}$ in which 4 UAVs and 2 macrocell base stations are deployed based on a traditional grid-based deployment.
 The ground users are distributed according to a truncated Gaussian distribution with a standard deviation $\sigma_o$. This type of distribution which is suitable to model a hotspot area. %, is given by:\vspace{-0.3cm}
% \begin{equation}
% f(x,y) = \frac{1}{\eta}{\rm{exp}}{\left( {\frac{{{L} - {\mu _x}}}{{\sqrt {2{\sigma _x}} }}} \right)^2}{\rm{exp}}{\left( {\frac{{{L} - {\mu _y}}}{{\sqrt {2{\sigma _y}} }}} \right)^2},\vspace{-0.1cm}
% \end{equation}
%where $\eta= 2\pi {\sigma _x}{\sigma _y}{\rm{erf}}\left( {\frac{{{L} - {\mu _x}}}{{\sqrt {2{\sigma _x}} }}} \right){\rm{erf}}\left( {\frac{{{L} - {\mu _y}}}{{\sqrt {2{\sigma _y}} }}} \right)$, and the size of the area is $L\times L$. Also, ${\mu _x}$, ${\sigma _x}$, ${\mu _y}$, and ${\sigma _y}$ are the mean and standard deviation values of the $x$ and $y$ coordinates, and ${\rm{erf(}}z{\rm{)}} = \frac{2}{{\sqrt \pi  }}\int\limits_0^z {{e^{ - {t^2}}}{\rm{d}}t}$.
%In this case, (${\mu _x}$, ${\mu _y}$) represents the center of the hotspot, and the density of the users around the center is inversely proportional to the values ${\sigma _x}$ and ${\sigma _y}$. 
%Table I lists the simulation parameters. 
%Here, to highlight the gains from using UAVs, we compare our results with a case in which small ground BSs are used to support the macrocell BSs (macro-BSs). 
\textcolor{black}{The simulation parameters are given as follows. $f_c=$2\,GHz, transmit power of each BS is 40\,W, and transmit power of each UAV is 1\,W. Also, $N=300$, $W_j=W_i=1\,\textrm{MHz}$, and the noise power spectral density is -170\,dBm/Hz. We consider a dense urban environment with $n=3$,  	$\mu_\textrm{LoS}=3\,\textrm{dB}$, $\mu_\textrm{NLoS}=23\,\textrm{dB}$, $\alpha=0.36$, and $\gamma=0.21$ \cite{HouraniModeling}. The heights of each UAV and BS are, respectively, 200\,m and 20\,m \cite{HouraniModeling, Vishal,mozaffari2}. All statistical results are averaged over a large number of independent runs.}

%\begin{figure}[!t]
%	\begin{center}
%		\vspace{-0.1cm}
%		\includegraphics[width=8.0cm]{./Plots/Range_vs_H_new.eps}
%		\vspace{-0.2cm}
%		\caption{ \small Coverage range versus UAV's altitude.\vspace{-.32cm}}
%		\label{Altitude}
%	\end{center}
%\end{figure} 

%Figure \ref{Altitude} shows the coverage range of the UAV as its altitude varies. Clearly, as shown in Proposition 1, there is an optimal altitude that leads to a maximum coverage range. In this case, for altitudes higher than the optimal one, a lower coverage range is obtained due to the higher path loss resulting from the higher distance. For lower altitudes, however, the lower LoS probability leads to a higher path loss, and hence, lower coverage range. As we can see from Figure \ref{Altitude}, the optimal altitudes of the UAV is around 1000\,m that leads to 1500\,m coverage range. However, for 2\,MHz bandwidth, the coverage range decreases to 1100\,m (for 700\,m altitude) due to the higher noise power in higher bandwidths. As a result, as the bandwidth increases, the UAV must reduce its altitude to mitigate the path loss.  

%\begin{figure}[!t]
%	\begin{center}
%		\vspace{-0.1cm}
%		\includegraphics[width=5.5cm]{./Plots/Sigma.eps}
%		\vspace{-0.25cm}
%		\caption{ \small Average rate comparison between the proposed cell association and the SNR-based association. }\vspace{-1.00cm}
%		\label{Sigma}
%	\end{center}
%\end{figure} 

\textcolor{black}{In Fig.\,\ref{Delay}, we compare the delay of the proposed cell association with the traditional SNR-based association. We consider a truncated Gaussian distribution with a center (1300\,m,1300\,m), and $\sigma_o$ varying from 200\,m to 1200\,m. Lower values of $\sigma_o$ correspond to scenarios in which users are more concentrated around the hotspot center. Fig.\,\ref{Delay} shows  that the proposed cell association significantly outperforms the SNR-based association in terms of the average delay. For low $\sigma_o$ values, the average delay decreases by 72\% compared to the SNR-based association. This is due to the fact that, in the proposed approach, the impact of network congestion is taken into account. %In other words, while determining the cell partitions, impacts of both the average number of users in each cell and the SNR values are considered.
Hence, the proposed approach avoids creating highly loaded cells. In contrast, an SNR-based association can yield highly loaded cells. As a result, in the congested cells, each user will receive a low amount of bandwidth that leads a low transmission rate or equivalently high delay. In fact, compared to the SNR-based association case, our approach is more robust against network congestion and its performance is significantly less affected by changing $\sigma_o$.} %Consequently,  thus reducing the users' throughput due the bandwidth limitation.} %Clearly, as $\sigma_o$ increases, the average rate performance of the SNR-based method becomes closer to the proposed approach. This verifies that, the SNR-based association is optimal when users are uniformly distributed.% over the~area.      

\begin{figure}[!t]
	\begin{center}
		\vspace{-0.2cm}
		\includegraphics[width=7cm]{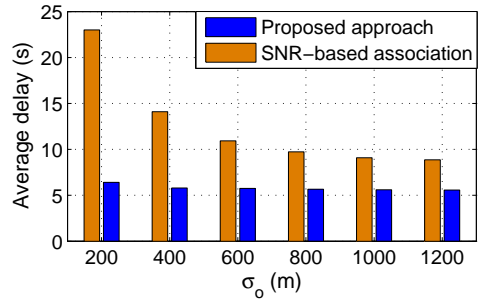}
		\vspace{-0.15cm}
		\caption{ \small Average network delay per 1Mb data transmission.\vspace{-.1cm}}
		\label{Delay}
	\end{center}
\end{figure}

%\begin{figure}[!t]
% \begin{minipage}{.49\textwidth}
%  \begin{center}
%   \vspace{-0.2cm}
%    \includegraphics[width=0.86\textwidth]{./Plots_latex/Convergence.eps}
%    \vspace{-0.5cm}
%    \caption{ \small Overall convergence of the algorith.\vspace{0.1cm}}
%    \label{Convergence}
%  \end{center}\vspace{-1.41cm}
% \end{minipage}
% \hfill
% \begin{minipage}{.49\textwidth}
%  \begin{center}
%   \vspace{-0.2cm}
%    \includegraphics[width=0.86\textwidth]{./Plots_latex/ComparedToOptimal.eps}
%    \vspace{-0.5cm}
%    \caption{ \small Comparison between the proposed approach\vspace{-0.27cm}\\ and the optimal solution.\vspace{-0.2cm}}
%    \label{ComparedwithOptimal}
%  \end{center}\vspace{-1.41cm}
%   \end{minipage}
%\end{figure}

\begin{figure}[t]
	\centering
	\hspace{-1.10cm}
	\begin{subfigure}[t]{0.2\textwidth}
		\begin{center}
			\vspace{-0.3cm}
			\includegraphics[width=4.6cm]{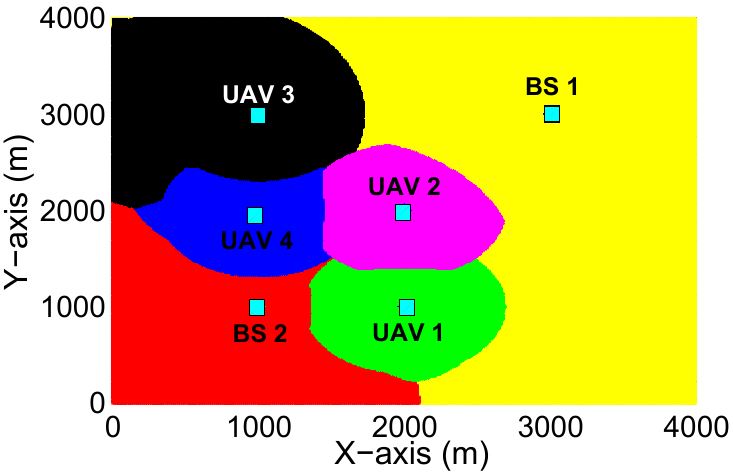}
			\vspace{-0.5cm}
			\caption{ \textcolor{black}{\footnotesize Proposed delay-optimal cell partitions.}\vspace{-.02cm}}
			\label{Ayy}
		\end{center}
	\end{subfigure}%
	~	\hspace{0.69cm}
	\begin{subfigure}[t]{0.2\textwidth}
		\begin{center}
			\vspace{-0.3cm}
			\includegraphics[width=4.6cm]{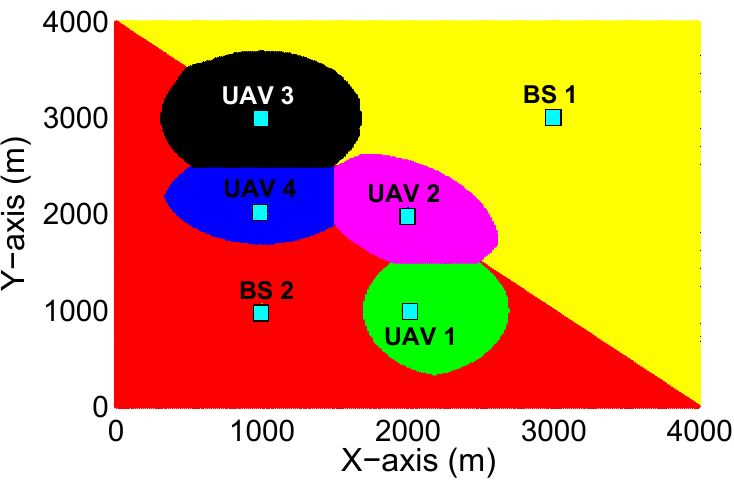}
			\vspace{-0.5cm}
			\caption{ \textcolor{black}{\footnotesize SNR-based association.} \vspace{-.02cm}}
			\label{Voro}
		\end{center}
	\end{subfigure}\vspace{-0.1cm}	
\caption{ \textcolor{black}{\small Cell partitions associated to UAVs and BSs given the non-uniform spatial distribution~of~users.} \vspace{-.18cm} 	\label{CellPart}}\vspace{-0.4cm} 
\end{figure}

\textcolor{black}{As an illustrative example, Fig.\,\ref{CellPart} shows the locations of the BSs and UAVs as well as the cell partitions
obtained using SNR-based association and the proposed delay-optimal
association. In this case, users are distributed based on a 2D truncated Gaussian distribution with mean values of (1300\,m,1300\,m), and $\sigma_o=1000\,\textrm{m}$. As shown in Fig.\,\ref{CellPart}, the size and shape of cells are different in these two association approaches. For instance, the
red cell partition in the proposed approach is smaller than the
SNR-based case. In fact, the red partition in the SNR-based
approach is highly congested and, consequently, its size is
reduced in the proposed case so as to decrease the congestion as well as the delay.} \vspace{-0.02cm}

\section{Conclusion}\vspace{-0.05cm}
In this paper, we have proposed a novel framework for delay-optimal cell association in UAV-enabled cellular networks. % First, we have determined the maximum achievable coverage range of UAVs and terrestrial BSs.
 In particular, to minimize the average network delay based on the users' distribution, we have exploited optimal transport theory to derive the optimal cell associations for UAVs and terrestrial BSs. \textcolor{black}{Our results have shown that, the proposed cell association approach results in a significantly lower network delay compared to an SNR-based association.} \vspace{-0.09cm} 

\bibliographystyle{IEEEtran}

\bibliography{references}

% Generated by IEEEtran.bst, version: 1.13 (2008/09/30)
\begin{thebibliography}{10}
\providecommand{\url}[1]{#1}
\csname url@samestyle\endcsname
\providecommand{\newblock}{\relax}
\providecommand{\bibinfo}[2]{#2}
\providecommand{\BIBentrySTDinterwordspacing}{\spaceskip=0pt\relax}
\providecommand{\BIBentryALTinterwordstretchfactor}{4}
\providecommand{\BIBentryALTinterwordspacing}{\spaceskip=\fontdimen2\font plus
\BIBentryALTinterwordstretchfactor\fontdimen3\font minus
  \fontdimen4\font\relax}
\providecommand{\BIBforeignlanguage}[2]{{%
\expandafter\ifx\csname l@#1\endcsname\relax
\typeout{** WARNING: IEEEtran.bst: No hyphenation pattern has been}%
\typeout{** loaded for the language `#1'. Using the pattern for}%
\typeout{** the default language instead.}%
\else
\language=\csname l@#1\endcsname
\fi
#2}}
\providecommand{\BIBdecl}{\relax}
\BIBdecl

\bibitem{orfanus}
D.~Orfanus, E.~P. de~Freitas, and F.~Eliassen, ``Self-organization as a
  supporting paradigm for military {UAV} relay networks,'' \emph{IEEE
  Communications Letters}, vol.~20, no.~4, pp. 804--807, 2016.

\bibitem{mozaffari2}
M.~Mozaffari, W.~Saad, M.~Bennis, and M.~Debbah, ``Unmanned aerial vehicle with
  underlaid device-to-device communications: Performance and tradeoffs,''
  \emph{IEEE Transactions on Wireless Communications}, vol.~15, no.~6, pp.
  3949--3963, June 2016.

\bibitem{zhangLetter}
J.~Lyu, Y.~Zeng, and R.~Zhang, ``Cyclical multiple access in uav-aided
  communications: A throughput-delay tradeoff,'' \emph{IEEE Wireless
  Communications Letters}, vol.~5, no.~6, pp. 600--603, 2016.

\bibitem{zhang}
Y.~Zeng, R.~Zhang, and T.~J. Lim, ``Wireless communications with unmanned
  aerial vehicles: opportunities and challenges,'' \emph{IEEE Communications
  Magazine}, vol.~54, no.~5, pp. 36--42, May 2016.

\bibitem{Letter}
M.~Mozaffari, W.~Saad, M.~Bennis, and M.~Debbah, ``Efficient deployment of
  multiple unmanned aerial vehicles for optimal wireless coverage,'' \emph{IEEE
  Communications Letters}, vol.~20, no.~8, pp. 1647--1650, Aug. 2016.

\bibitem{Irem}
R.~Yaliniz, A.~El-Keyi, and H.~Yanikomeroglu, ``Efficient 3-{D} placement of an
  aerial base station in next generation cellular networks,'' in \emph{Proc. of
  IEEE International Conference on Communications (ICC)}, Kuala Lumpur,
  Malaysia, May 2016.

\bibitem{HouraniModeling}
A.~Hourani, S.~Kandeepan, and A.~Jamalipour, ``Modeling air-to-ground path loss
  for low altitude platforms in urban environments,'' in \emph{Proc. of IEEE
  Global Communications Conference (GLOBECOM)}, Austin, TX, USA, Dec. 2014.

\bibitem{Jeong}
S.~Jeong, O.~Simeone, and J.~Kang, ``Mobile edge computing via a {UAV}-mounted
  cloudlet: Optimal bit allocation and path planning,'' \emph{IEEE Transactions
  on Vehicular Technology}, to appear, 2017.

\bibitem{Qing}
Q.~Wu, Y.~Zeng, and R.~Zhang, ``Joint trajectory and communication design for
  multi-{UAV} enabled wireless networks,'' \emph{available online:
  arxiv.org/abs/1705.02723}, 2017.

\bibitem{HoverTime}
M.~Mozaffari, W.~Saad, M.~Bennis, and M.~Debbah, ``Wireless communication using
  unmanned aerial vehicles ({UAVs}): Optimal transport theory for hover time
  optimization,'' \emph{available online: arxiv.org/abs/1704.04813}, 2017.

\bibitem{Vishal}
V.~Sharma, M.~Bennis, and R.~Kumar, ``{UAV}-assisted heterogeneous networks for
  capacity enhancement,'' \emph{IEEE Communications Letters}, vol.~20, no.~6,
  pp. 1207--1210, June 2016.

\bibitem{Alonso}
A.~Silva, H.~Tembine, E.~Altman, and M.~Debbah, ``Optimum and equilibrium in
  assignment problems with congestion: Mobile terminals association to base
  stations,'' \emph{IEEE Transactions on Automatic Control}, vol.~58, no.~8,
  pp. 2018--2031, Aug. 2013.

\bibitem{OTUAV}
M.~Mozaffari, W.~Saad, M.~Bennis, and M.~Debbah, ``Optimal transport theory for
  power-efficient deployment of unmanned aerial vehicles,'' in \emph{Proc. of
  IEEE International Conference on Communications (ICC)}, May 2016.

\bibitem{villani}
C.~Villani, \emph{Topics in optimal transportation}.\hskip 1em plus 0.5em minus
  0.4em\relax American Mathematical Soc., 2003, no.~58.

\bibitem{Crippa}
G.~Crippa, C.~Jimenez, and A.~Pratelli, ``Optimum and equilibrium in a
  transport problem with queue penalization effect,'' \emph{Advances in
  Calculus of Variations}, vol.~2, no.~3, pp. 207--246, 2009.

\end{thebibliography}
%\vspace{-0.4cm}
% that's all folks
\end{document}